\documentclass[11pt,letterpaper]{article}
\usepackage[T1]{fontenc}
\usepackage{lmodern,microtype}
\usepackage{amsmath,amssymb,amsthm,mathtools}
\usepackage[margin=1in]{geometry}
\usepackage{enumitem,needspace}
\usepackage[colorlinks=true,linkcolor=black,citecolor=black,urlcolor=blue]{hyperref}
\hypersetup{pdftitle={Near-Inverse-Linear Barriers for Explicit Affine Witness Isolation},pdfauthor={Sebastian Ben Daniel}}
\pdftrailerid{}
\newtheorem{theorem}{Theorem}[section]
\newtheorem{lemma}[theorem]{Lemma}
\newtheorem{proposition}[theorem]{Proposition}
\newtheorem{corollary}[theorem]{Corollary}
\theoremstyle{definition}

\theoremstyle{remark}

\newcommand{\F}{\mathbb F}
\newcommand{\K}{\mathbb K}
\newcommand{\Sol}{\operatorname{Sol}}
\newcommand{\rank}{\operatorname{rank}}
\newcommand{\im}{\operatorname{im}}
\newcommand{\codim}{\operatorname{codim}}
\newcommand{\Tr}{\operatorname{Tr}}
\newcommand{\NP}{\mathsf{NP}}
\newcommand{\Ppoly}{\mathsf{P}/\mathsf{poly}}

\newcommand{\USAT}{\mathsf{prUSAT}}
\newcommand{\Fam}{\mathcal K}
\newcommand{\Iso}{\mathcal I}
\newcommand{\one}{\mathbf 1}
\newcommand{\doi}[1]{\href{https://doi.org/#1}{doi:\nolinkurl{#1}}}
\setlist{itemsep=3pt,topsep=5pt,leftmargin=1.6em}
\allowdisplaybreaks[1]
\numberwithin{equation}{section}
\title{Near-Inverse-Linear Barriers for\\Explicit Affine Witness Isolation}
\author{Sebastian Ben Daniel}
\date{October 2026}
\begin{document}
\maketitle
\begin{abstract}
The Valiant--Vazirani isolation lemma reduces the satisfying
assignments of a circuit on $n$ variables to exactly one with
probability $\Omega(1/n)$, by intersecting them with a random affine
subspace chosen without looking at the circuit. We ask how much an
isolator can gain by reading the circuit while keeping the same form
of output. Our isolators are randomized nonuniform polynomial-size
transformations that receive a circuit and output the coefficients
of a system of linear equations over $\F_2$; they succeed when
exactly one satisfying assignment obeys the system. Success is
required only on circuits whose satisfying set is a nonempty affine
subspace, and no basis of that subspace is given. We show that
success probability $\Omega(n^{-\delta})$, for any fixed $\delta<1$,
implies $\NP\subseteq\Ppoly$. Conversely, $\NP\subseteq\Ppoly$ yields
a deterministic isolator of this kind, so the two statements are
equivalent. Closer to $1/n$, any success probability
$\omega(\log n/n)$ implies circuits of size $(s+2)^{O(1)}2^{o(v)}$
for the satisfiability of length-$s$ circuits on $v$ variables. The
results extend to outputs that are unions of at most $n^\beta$ affine
systems, with success $\Omega(n^{-\delta})$, whenever
$\beta+\delta<1$. We obtain no consequence from success $C/n$ with
$C$ constant.
\end{abstract}

\section{Introduction}\label{sec:result}

Valiant and Vazirani \cite{VV} showed how to turn a satisfiable
circuit into a uniquely satisfiable one. Given a circuit $H$ on $n$
variables, choose a random number of random linear equations over
$\F_2$ and add them to $H$ as extra constraints. With probability
$\Omega(1/n)$ exactly one satisfying assignment of $H$ survives.
The equations are chosen without looking at $H$.

Dell, Kabanets, van Melkebeek, and Watanabe \cite{DKMW} asked
whether the success probability can be improved. They considered
arbitrary randomized nonuniform polynomial-size procedures that map
$H$ to a circuit whose satisfying assignments are among those of $H$,
and showed that success probability above $2/3$ is achievable if and
only if $\NP\subseteq\Ppoly$. The companion paper \cite{BD} extends
the collapse to constant success probabilities. For a natural class
of blackbox procedures that includes the Valiant--Vazirani one, the
authors of \cite{DKMW} also proved an unconditional upper bound of
$O(1/n)$ on the success probability.

So for procedures that may read the circuit, a collapse is known at
constant success probability, while the best algorithm succeeds with
probability about $1/n$. This paper narrows the gap for the kind of
output that Valiant and Vazirani use: a system of linear equations.

\paragraph{The model.}
An \emph{explicit affine isolator} receives the description of a
Boolean circuit $H$ on $n$ ordered witness variables and outputs a
binary matrix $A$ with $n$ columns and a vector $b$. We call the
system $Az=b$ a \emph{filter}. The isolator succeeds on $H$ if exactly
one satisfying assignment of $H$ satisfies $Az=b$. The isolator is a
randomized nonuniform circuit family of size polynomial in the
description length $L$ of $H$. The matrix, the right-hand side, and
the number of equations may depend on the whole description and on
the random seed.

We require success only on \emph{affine inputs}: circuits whose
satisfying set $\Sol(H)$ is a nonempty affine subspace of $\F_2^n$.
Such a circuit is still an arbitrary circuit. It need not display
equations, an origin, or a basis for its satisfying set. Asking for
success on fewer inputs makes the hypothesis weaker and the result
stronger. If the dimension $r$ of $\Sol(H)$ were known, $r$ random
equations would isolate with constant probability. The difficulty
for the isolator is that $r$ is hidden in the syntax of $H$.

\subsection{Results}

\begin{theorem}[Fixed-exponent affine isolation barrier]\label{thm:main}
Fix an integer $k\ge1$, and put $\alpha=k/(k+1)$. Let $\Iso$ be a
randomized nonuniform polynomial-size transformation. On a circuit
$H$ of description length $L$, with $n$ ordered witness variables,
it returns an explicit binary matrix $A$ with $n$ columns and a
compatible vector $b$. Its computation, advice, random-seed length,
and output length are bounded by fixed polynomials in $L$.
The matrix, right-hand side, and number of rows may depend on the
whole description of $H$ and on the seed.

Suppose that, for every sufficiently large $n$ and every such circuit
with a nonempty affine satisfying set,
\begin{equation}\label{eq:main}
 \Pr\bigl[\,|\Sol(H)\cap\{z:Az=b\}|=1\,\bigr]
 \ge \frac{2}{n^{\alpha}}.
\end{equation}
Then $\NP\subseteq\Ppoly$.
\end{theorem}

\begin{corollary}[Every fixed exponent below one]\label{cor:exponents}
In the same model, a worst-case singleton guarantee
$\Omega(n^{-\delta})$ on every sufficiently large nonempty affine
circuit input, for any fixed $\delta<1$, implies
$\NP\subseteq\Ppoly$.
\end{corollary}
\begin{proof}
Write the guarantee as at least $c n^{-\delta}$ for a fixed $c>0$.
Choose a fixed positive integer $k$ with $k/(k+1)>\delta$.
Then $c n^{-\delta}\ge2n^{-k/(k+1)}$ for all sufficiently large $n$,
and Theorem~\ref{thm:main} applies. If $\delta<0$, the stated
asymptotic probability guarantee itself is impossible; the
nonvacuous range is $0\le\delta<1$.
\end{proof}

The exponent is measured in the witness arity $n$, not in the
description length $L$; the two are kept apart throughout. Random
affine hashing achieves $\Omega(1/n)$ on every satisfiable input
(Section~\ref{sec:VV}), so the exponent $1$ cannot be improved unless
$\NP\subseteq\Ppoly$.

The implication is an equivalence. If $\NP\subseteq\Ppoly$, the usual
search-to-decision reduction finds a satisfying assignment $w$, and
the filter $z=w$ has success one. Hence either some polynomial-size
explicit affine isolator has success one, or none has success
$\Omega(n^{-\delta})$ for any fixed $\delta<1$
(Corollary~\ref{cor:equivalence}).

The construction has an integer parameter $k$, which is fixed in
Theorem~\ref{thm:main}. Letting $k$ grow with the input gives
consequences for success probabilities closer to $1/n$. The gadget
becomes larger, so the conclusions are weaker than
$\NP\subseteq\Ppoly$. We write $\log$ for base-two logarithms and
$\ln$ for natural logarithms. An isolator has \emph{resource
degree} $a$ if its circuit size, advice, seed length, and output
length are at most $(L+2)^a$.

\begin{theorem}[Growing-parameter consequences]\label{thm:growingintro}
In the model of Theorem~\ref{thm:main}, suppose the stated success
lower bound holds on every sufficiently large nonempty affine
circuit input.
\begin{enumerate}
\item If success is at least $c\,2^{\theta\sqrt{\log_2 n}}/n$ for
fixed $\theta,c>0$, then every language in $\NP$ has circuits of size
$2^{O((\log s)^2)}$ on inputs of length $s$.
\item If success is at least $c(\log_2 n)^A/n$ for fixed $c>0$ and
$A>1$, then circuit satisfiability on descriptions of length $s$
has circuits of size $2^{O(s^{1/A})}$, and hence subexponential size.
\item If success is at least $p(n)$ with $np(n)/\log_2n\to\infty$,
then circuit satisfiability on length-$s$ descriptions with at most
$v$ witness variables has circuits of size
$(s+2)^{O(1)}2^{o(v)}$, for $1\le v\le s$.
\item If success is at least $C\log_2(n)/n$ for a fixed $C>0$,
and the isolator has resource degree $a$, then circuit
satisfiability on length-$s$ descriptions with at most $v$ witness
variables has circuits of size
$(s+2)^{O(1)}2^{(2\kappa(a)/C)v}$. Here $\kappa(a)$ is the
parameter-independent resource exponent of
Theorem~\ref{thm:resource}.
\end{enumerate}
\end{theorem}

The constants in these bounds may depend on the isolator and on the
fixed parameters shown. Part~2 is stated for circuit satisfiability;
for another language in $\NP$, the length of its reduction to circuit
satisfiability enters the exponent (Section~\ref{sec:growing}). In
part~4 the rate is proportional to $\kappa(a)/C$, so a large constant
$C$ helps only in relation to the isolator's resource degree. All
four conclusions are nonuniform circuit upper bounds. Ruling out the
hypotheses therefore needs nonuniform hardness assumptions; the
uniform exponential time hypothesis is not enough. We obtain no
consequence from success $C/n$ with $C$ fixed.

Finally, the argument extends to filters that are unions of several
affine systems, each given by its coefficients. If the number of
systems is at most $n^\beta$ and the success probability is
$\Omega(n^{-\delta})$ with $\beta+\delta<1$, then again
$\NP\subseteq\Ppoly$ (Corollary~\ref{cor:unions}). This includes
parity decision trees with at most $n^\beta$ accepting leaves.

\subsection{Proof overview}\label{sec:overview}

The proof follows the strategy of \cite{DKMW}: use the isolator to
build small nonuniform circuits for unique satisfiability. Let
$\USAT$ be the promise problem whose YES instances are circuits with
exactly one satisfying assignment and whose NO instances are
unsatisfiable circuits. By \cite{VV}, polynomial-size circuits for
$\USAT$ give $\NP\subseteq\Ppoly$.

\paragraph{A gadget that hides a dimension.}
Take $m$ instances $E_1,\ldots,E_m$ of $\USAT$, the \emph{sources},
each on $q$ variables, and fix affine sets
$K_1,\ldots,K_m\subseteq\F_2^d$, the \emph{cells}. We build one
circuit $H_Q$ on variables $z=(t,x_1,\ldots,x_m)$, with a \emph{tag}
$t\in\F_2^d$ and one $q$-bit block $x_i$ per source, such that:
\begin{itemize}
\item if every source has a unique witness, $\Sol(H_Q)$ is the graph
of an affine map $t\mapsto(x_1,\ldots,x_m)$ over all of $\F_2^d$;
\item if $E_j$ is unsatisfiable and the other sources have unique
witnesses, $\Sol(H_Q)$ is a smaller affine set $R_j$ lying over the
cell $K_j$.
\end{itemize}
When all sources are satisfiable, $R_j$ is exactly the part of the
graph over $K_j$, whatever the witness of $E_j$ is. The description
of $H_Q$ is computed from the sources alone and is the same in both
cases, so an isolator that does well on inputs of the second kind
must also isolate inside the regions $R_j$ on inputs of the first
kind, unless it can tell the two apart.

\paragraph{Testing one source.}
Suppose we know the witnesses of all sources except $E_j$. Then we
can compute $R_j$, by an origin and a basis, without knowing whether
$E_j$ is satisfiable. Run the isolator on $H_Q$ with a fixed list of
seeds and count how often its filter meets $R_j$ in exactly one
point. Since the filter is given by its coefficients, this is a rank
computation. If $E_j$ is unsatisfiable, then $R_j=\Sol(H_Q)$ and the
count is at least the isolator's success probability, up to a
sampling error. If all sources are satisfiable, each filter cuts the
graph in an affine set, whose projection $T$ to the tag space is
again affine, and the filter isolates inside $R_j$ exactly when
$|T\cap K_j|=1$. The same $T$ serves all $j$, because the description
of $H_Q$ does not depend on $j$.

\paragraph{The flag partition.}
How many cells can one affine set $T$ meet in exactly one point? For
the cells
\[
 K_j=\{t:t_1=\cdots=t_{j-1}=0,\ t_j=1\}\quad(1\le j\le d),
 \qquad K_{d+1}=\{0\},
\]
the answer is at most two (Lemma~\ref{lem:flag}). So with $m=d+1$
cells, in the all-YES case the count averaged over $j$ is at most
$2/m$. If the isolator's success exceeds $2/m$ by an inverse
polynomial, then in every pool of satisfiable sources the count is
low for at least one $j$, while it is never low when $E_j$ is
unsatisfiable. The domination argument of \cite{DKMW}
(Lemma~\ref{lem:domination}) turns such one-sided tests into a
polynomial-size decider for $\USAT$.

\paragraph{Keeping the arity small.}
The threshold $2/m$ has to be expressed in the arity $n=d+mq$ of
$H_Q$. A direct encoding, in which block $i$ records the defining
equations of $K_i$ injectively, needs $q$ to be at least the
codimension of $K_i$, which is as large as $d$. That gives
$n\approx m^2$ and the threshold $2/\sqrt n$. To do better we pack
the equations of $K_i$ into the coefficients of a polynomial
$P_i(t;W)$ over $\F_{2^q}$ of degree about $d/q$, with no constant
term, and let the circuit accept the block value $x_i$ when some root
$w$ of $P_i(t;W)-x_i$ satisfies $E_i$. All roots can be listed by a
deterministic circuit, so no new variables are needed. For a source
with unique witness $w_i$ this forces $x_i=P_i(t;w_i)$, which is
affine in $t$. Now $q$ can be the witness length of the sources while
$d$ is any fixed power of it. With
\begin{equation}\label{eq:parameters}
 1\le q\le\ell,\qquad m=(q+1)^k,\qquad d=m-1,
 \qquad n=d+mq=(q+1)^{k+1}-1,
\end{equation}
where $\ell$ is the description length of the sources, the bound
$2/m$ equals $2/(n+1)^{k/(k+1)}$, which is strictly below
\eqref{eq:main}. Theorem~\ref{thm:main} takes $q=\ell$; keeping $q$
separate from $\ell$ matters for Theorem~\ref{thm:growingintro}.

\paragraph{Why not $1/n$.}
Any gadget of this type with one $q$-bit block per source has
threshold at least $2/m\ge2q/n$, and the flag construction attains
this up to a factor $1+1/q$ (Corollary~\ref{cor:blockbarrier}). The
remaining distance from $1/n$ is the witness length $q$ of the
sources. Raising $k$ makes $q$ smaller relative to $n$, but once $k$
grows the gadget is larger than any fixed polynomial in the source
length. This is why Theorem~\ref{thm:growingintro} has weaker
conclusions.

\subsection{Scope and related work}

The filter must be given by its coefficients. The test in the proof
decides by Gaussian elimination whether a filter meets a known affine
region in exactly one point, and a circuit that is merely promised to
compute an affine predicate does not give access to the coefficients
this needs. For the same reason, the extension to unions requires
each affine system to be listed. Whether the results hold for
arbitrary output circuits is open; the collapses known in that
setting \cite{DKMW,BD} are far from $1/n$.

The proof combines several known ingredients. The overall strategy,
including the passage from one-sided testers to nonuniform circuits,
follows \cite{DKMW}, and the reduction from satisfiability to unique
satisfiability is that of \cite{VV}. The partition of
Section~\ref{sec:flag} is the binary case of a recursive construction
of Bamberg, Filmus, Ihringer, and Kurz \cite[Lemma~6, p.~335]{BFIK};
we give it in coordinates and prove the property we use. Root listing
in Section~\ref{sec:roots} is Berlekamp's deterministic trace method
\cite{Ber67,Ber70}, proved here in the special case we need. The
evaluation map of Section~\ref{sec:compiler} has the form of the
polynomial hash without constant term described by Nguyen and Roscoe
\cite[Section~5, p.~9]{NR}. The search-to-decision reduction behind
the converse is standard; see for example \cite[Section~1.1]{BG}.
The paper is self-contained: every step needed for
Theorem~\ref{thm:main} is proved below.

\paragraph{Organization.}
Section~\ref{sec:model} fixes the encoding and the parameters.
Sections~\ref{sec:flag}--\ref{sec:size} build the gadget: the flag
partition, root listing, the compiled circuit, and its size.
Sections~\ref{sec:seeds}--\ref{sec:VV} prove Theorem~\ref{thm:main}.
Section~\ref{sec:growing} proves Theorem~\ref{thm:growingintro},
Section~\ref{sec:converse} the converse, and
Section~\ref{sec:unions} the extension to unions.
Section~\ref{sec:barriers} shows that the threshold $2/m$ cannot be
improved within this method.

\section{Model and conventions}\label{sec:model}

This section fixes the circuit encoding, the source promise problem,
and the parameters. Table~\ref{tab:parameters} at the end of the
section lists the notation.

All vector spaces and affine systems are over $\F_2$ unless another
field is specified. An affine set is a translate of a linear
subspace. We also allow the empty set when maximizing over retained
sets. A nonempty affine set of dimension $r$ has $2^r$ points; in
particular, every pair of distinct points is a binary affine line.

Fix an indexed bounded-fan-in Boolean-circuit encoding in which a
valid description of length $\ell$ has at most $\ell$ ordered witness
variables and $O(\ell)$ gates. Parsing and evaluation are polynomial
time. Allow a syntactically explicit ignored padding field, so a
constructed circuit can be padded to a common prescribed length
without changing its function or ordered variables. An encoding
with efficiently attainable common lengths and the same asymptotic
bounds is equally suitable.

The source promise problem $\USAT$ has YES instances with exactly
one satisfying assignment and NO instances with none. The decider
may answer arbitrarily on circuits with at least two satisfying
assignments. Malformed descriptions are rejected. At source length
$\ell\ge1$, fix a witness cap $1\le q\le\ell$. We consider the
restricted promise problem $\USAT(\ell,q)$ consisting of length-$\ell$
circuit descriptions with at most $q$ witness variables. Normalize
a circuit on $v\le q$ variables
to $q$ variables by conjoining $w_{v+1}=\cdots=w_q=0$.
This preserves both zero and unique satisfiability and adds $O(q)$
gates. Thus the $q$-bit strings can serve uniformly as candidate
source witnesses.

Fix a degree-$q$ irreducible polynomial $M_q(X)$ over $\F_2$ and a
polynomial-basis representation
\[
 \K=\F_{2^q}=\F_2[X]/(M_q(X)).
\]
Finite-field existence supplies such a modulus. Its $O(q)$
coefficients can be chosen as length-dependent nonuniform data;
finding them uniformly is unnecessary. Fix the order of the basis,
so every field element has one specified $q$-bit representation.

A \emph{pool} $Q$ consists of $m$ distinct source descriptions of
length $\ell$ with at most $q$ witness variables. Sort them
lexicographically and let $E_1,\ldots,E_m$ be their normalized
predicates. Distinct descriptions may compute the same function or
have the same unique witness. A pool is \emph{all-YES} if every
member is a YES instance, and \emph{one-NO} if exactly one member is
a NO instance and the others are YES instances. In the argument, one
member of the pool is \emph{designated} as the instance under test,
and the witnesses of the other $m-1$ members are known. The compiler
of Section~\ref{sec:compiler} receives only the pool and fixed
length-dependent data. It receives neither the witnesses nor the
designation.

For input length $L$, regard the isolator as a Boolean circuit with
its length-$L$ advice incorporated and a fixed polynomial random-bit
budget. Unused random bits are ignored. Variable output row counts
can be represented by a count and a padded coefficient array. Thus a
particular input description and seed determine exactly one explicit
filter. As in Section~\ref{sec:result}, the isolator has resource
degree $a$ if its circuit size, advice, seed length, and output
length are at most $(L+2)^a$.

We do not assume that the isolator treats two descriptions of the
same function alike. For this reason the compiler will produce, for
all pools at the same $(\ell,q,k)$, descriptions of one common length
$L_{\ell,q,k}=(n+\ell+2)^{O(1)}$, with an exponent independent of
$k$. We write $L$ or $L_\ell$ for this length once the schedule for
$q$ and $k$ has been fixed.

\begin{table}[ht]
\centering
\renewcommand{\arraystretch}{1.15}
\begin{tabular}{@{}lp{0.74\textwidth}@{}}
\hline
Symbol & Meaning\\
\hline
$\ell$ & description length of a source instance\\
$q$ & number of witness variables of a source, $1\le q\le\ell$;
      the field is $\K=\F_{2^q}$\\
$k$ & exponent parameter, with $\alpha=k/(k+1)$\\
$m=(q+1)^k$ & number of sources in a pool, and number of cells\\
$d=m-1$ & number of tag bits\\
$n=d+mq$ & witness arity of the compiled circuit, equal to
           $(q+1)^{k+1}-1$\\
$D=\lceil d/q\rceil$ & largest degree of a residual polynomial\\
$L$ & description length of the compiled circuit\\
$N=n+\ell+2$ & size parameter: all resources are polynomial in $N$\\
$a$, $\kappa(a)$ & resource degree of the isolator, and the exponent
                   of Theorem~\ref{thm:resource}\\
$s$, $v$ & description length and number of variables of a
           satisfiability instance (Sections~\ref{sec:VV}--\ref{sec:converse})\\
\hline
\end{tabular}
\caption{Parameters. The quantities $q,k,m,d,n$ are tied by
\eqref{eq:parameters}.}
\label{tab:parameters}
\end{table}

\section{The flag partition}\label{sec:flag}

The cells of the gadget come from a complete flag of subspaces of
$\F_2^d$. Lemma~\ref{lem:flag} is the only property of the cells used
in the proof of Theorem~\ref{thm:main}.

For $d\ge1$, put $V=\F_2^d$ and
\[
 F_0=V,\qquad F_j=\{t:t_1=\cdots=t_j=0\}
 \quad(1\le j\le d).
\]
Define $m=d+1$ cells by
\begin{align}
 K_j&=F_{j-1}\setminus F_j
     =\{t:t_1=\cdots=t_{j-1}=0,\ t_j=1\},
       &&1\le j\le d,\label{eq:cells}\\
 K_m&=F_d=\{0\}.\notag
\end{align}
Every nonzero vector has a unique first nonzero coordinate, so these
are a disjoint partition into nonempty proper affine sets.
Their codimensions are $1,2,\ldots,d,d$.
For any indexed sequence $\Fam=(K_1,\ldots,K_m)$ of regions, write
\begin{equation}\label{eq:f}
 f_{\Fam}(T)=\frac1m\sum_{i=1}^m\one\{|T\cap K_i|=1\}.
\end{equation}

\begin{lemma}[Two singleton intersections]\label{lem:flag}
For the partition \eqref{eq:cells}, every affine $T\subseteq V$,
including the empty set, meets at most two cells in exactly one
point. Consequently
\[
 \max_{T\text{ affine}} f_{\Fam}(T)=\frac2m.
\]
The maximum over singleton $T$ is $1/m$, and the maximum over
nonsingleton affine $T$ is $2/m$.
\end{lemma}
\begin{proof}
The empty case is immediate. Let $R_j=T\cap F_j$.
On a nonempty affine $R_{j-1}$, the coordinate $t_j$ is constant
zero, constant one, or nonconstant. In the first case
$T\cap K_j$ is empty. In the second, $T\cap K_j=R_{j-1}$ and
$R_j$ is empty. In the third, the two fibers have equal size:
translation by a direction vector on which $t_j=1$ bijects them.
Hence $|T\cap K_j|=|R_j|=|R_{j-1}|/2$.

At the first singleton intersection with a cell $K_j$, $j\le d$,
these alternatives imply that $R_j$ is empty or a singleton.
The later cells are disjoint and together contain exactly $R_j$
from $T$, so at most one later cell has a singleton intersection.
If none of the first $d$ cells has one, only $K_m$ can contribute.
This proves the upper bound two.

Choose points in two different cells. Their two-point affine line
meets exactly these cells, each once, proving attainment by a
nonsingleton set. A singleton belongs to precisely one cell.
\end{proof}

The proof uses only the successive equal-fiber property. It also
applies to a truncated flag: peel $m-1$ cells and use the residual
space as the last cell, for any $2\le m\le d+1$.

\section{Deterministic root listing over a binary field}\label{sec:roots}

The next lemma is the standard trace-splitting root algorithm,
specialized to $\F_2$-extensions; see Berlekamp
\cite[Section~5, especially equation~(5.02)]{Ber70}. We give a full
proof and fixed-layout circuit realization because the distinction
between computed candidates and witness variables is central here.
Its runtime is jointly polynomial in the field's bit length and
the formal degree bound.

\begin{lemma}[A fixed-layout root-list circuit]\label{lem:roots}
Given the fixed representation of $\K=\F_{2^q}$, all distinct roots
in $\K$ of a nonzero polynomial $f\in\K[W]$ of degree at most $D$
can be listed deterministically in time polynomial in $q$ and $D$.
There is a Boolean circuit of size polynomial in $q$ and $D$ whose
input is the coefficient vector and whose output has $D$ field-element
slots and $D$ validity bits. The valid slots contain each root exactly
once, in a fixed deterministic order.
\end{lemma}
\begin{proof}
We describe the algebraic algorithm and then its Boolean realization.
A nonzero constant polynomial has no roots, so suppose $\deg f\ge1$.

\paragraph{Extract the field-root factor.}
Compute the monic polynomial
\begin{equation}\label{eq:rootgcd}
 g(W)=\gcd\bigl(f(W),\,W^{2^q}-W\bigr).
\end{equation}
The exponentially large second polynomial need not be represented.
Start with $W\bmod f$, square modulo $f$ exactly $q$ times, and
subtract $W\bmod f$. This gives the remainder of $W^{2^q}-W$
modulo $f$, to which the Euclidean algorithm may be applied.
All stored remainders have degree below $D$.

Every element of $\K$ satisfies $w^{2^q}=w$. The monic polynomial
$W^{2^q}-W$ has degree $2^q$ and derivative $-1=1$, and is
therefore the product of the $2^q$ distinct linear factors
$W-w$, $w\in\K$. It follows that $g$ is exactly the squarefree
product of the linear factors of $f$ over $\K$. If $g=1$, return
the empty list. Otherwise $1\le\deg g\le D$.

\paragraph{Separate the roots by trace coordinates.}
Fix an $\F_2$-basis $\beta_1,\ldots,\beta_q$ of $\K$ and define
\[
 \Tr(u)=u+u^2+\cdots+u^{2^{q-1}}.
\]
This is $\F_2$-linear and takes values in $\F_2$, since squaring
permutes the displayed terms using $u^{2^q}=u$.
It is not the zero function: the displayed nonzero polynomial has
degree $2^{q-1}<2^q$ and cannot vanish at every field element.
For $v\ne0$, multiplication by $v$ permutes $\K$, so some $\beta$
satisfies $\Tr(\beta v)=1$. Expressing $\beta$ in the fixed basis
shows that $\Tr(\beta_jv)=1$ for some $j$.
Consequently the vectors
\begin{equation}\label{eq:tracecoordinates}
 \bigl(\Tr(\beta_1w),\ldots,\Tr(\beta_qw)\bigr)
\end{equation}
are distinct for distinct $w\in\K$.

Maintain a list of pairwise coprime monic nonconstant factors whose
product is $g$, initially the one-element list $(g)$.
For $j=1,\ldots,q$, compute
\[
 T_j(W)=\sum_{a=0}^{q-1}(\beta_jW)^{2^a}\pmod g
\]
by repeated modular squaring and addition. For every current factor
$h$, replace it by the nonconstant members of
\begin{equation}\label{eq:tracesplit}
 h_0=\gcd(h,T_j),\qquad h_1=h/h_0.
\end{equation}
The division is exact. As $h$ is squarefree and splits over $\K$,
$h_0$ consists precisely of its roots with $j$th trace coordinate
zero, and $h_1$ of its roots with that coordinate one. Reducing
$T_j$ modulo $g$ does not alter its value at any root of $h$.
At each stage there are at most $D$ nonconstant factors.
After all $q$ stages, \eqref{eq:tracecoordinates} implies that every
factor has at most one root, and hence has degree one. Read the
root from each factor. In characteristic two the root of
$c_1W+c_0$ is $c_0/c_1$; with monic factors it is simply $c_0$.
Sort the resulting $q$-bit representations to fix an output order.

\paragraph{Bound the computation and unroll it.}
Field addition is bitwise addition modulo two. Polynomial-basis
multiplication followed by reduction modulo $M_q$ uses polynomially
many bit operations in $q$. A nonzero field element has inverse
$u^{2^q-2}$, computable with $O(q)$ squarings and multiplications.
Polynomial addition, multiplication, long division, monic
normalization, and Euclidean gcd use polynomially many such
operations in $D$. Intermediate products have degree at most $2D$;
modular remainders and factors have degree at most $D$.
There are $q$ splitting stages with at most $D$ factors each,
and the trace residues use at most $q^2$ modular squarings.
These bounds are polynomial in $q,D$.

For a literal fixed-layout Boolean circuit, reserve $D$ factor
slots, validity flags, and $D+1$ coefficient slots per factor.
Represent degrees explicitly. Bound all arithmetic loops by their
polynomial worst-case lengths, implement branches by multiplexers,
and compact and sort valid slots by a fixed bounded procedure.
Division steps are guarded; zero denominators on inactive branches
can be replaced by one. This is a total deterministic bounded
computation, so unrolling it gives a polynomial-size Boolean circuit.
All intermediate coefficients, factors, roots, and flags are
computed wires. They are not additional existential inputs.
\end{proof}

The lemma does not assume $D<2^q$. A nonzero polynomial that vanishes
on all of $\K$ has $2^q\le D$ roots, and the list still fits. Repeated
roots, a zero formal derivative, the root zero, polynomials of degree
below $D$, and empty root lists are all handled by
\eqref{eq:rootgcd}. The zero polynomial is excluded; the compiler
treats that case separately before it uses the root list.

\section{Residual polynomials and the affine graph compiler}\label{sec:compiler}

This section defines the circuit $H_Q$ of the overview. Each source
$E_i$ gets one block $x_i\in\K$, and the block is tied to the tag $t$
through a polynomial whose coefficients record which defining
equations of the cell $K_i$ the tag violates.

Let $K_i\subseteq V=\F_2^d$ be a nonempty proper affine region
with $r_i$ independent defining equations
\[
 \lambda_{i,a}(t)=b_{i,a}\qquad(1\le a\le r_i).
\]
Write the residual bits as
$u_{i,a}(t)=\lambda_{i,a}(t)+b_{i,a}$.
Group them in their fixed order into
$D_i=\lceil r_i/q\rceil$ consecutive blocks of at most $q$ bits,
padding the last block with zero bits. Regard these blocks as field
elements $a_{i,1}(t),\ldots,a_{i,D_i}(t)$ and set
\begin{equation}\label{eq:residualpoly}
 P_i(t;W)=\sum_{h=1}^{D_i}a_{i,h}(t)W^h.
\end{equation}
Every coefficient is a binary affine function of $t$.
There is no constant coefficient. As a map from the coefficient
vector to an evaluation at $W=w$, \eqref{eq:residualpoly} is the
standard positive-degree polynomial hash \cite[Section~5, p.~9]{NR}.
Here its coefficients encode affine residuals and its evaluation
point is a candidate source witness.

\begin{lemma}[Residual encoding]\label{lem:residual}
The formal polynomial $P_i(t;W)$ is zero exactly when $t\in K_i$.
For $t\notin K_i$ it is formally nonconstant of degree at most $D_i$.
For every fixed $w\in\K$, the map $t\mapsto P_i(t;w)$ is
binary affine. For $t\notin K_i$ and any $x\in\K$, all solutions
of $P_i(t;w)=x$ can be listed by Lemma~\ref{lem:roots}.
\end{lemma}
\begin{proof}
All packed coefficients vanish exactly when every residual bit is
zero, which is equivalent to membership in $K_i$. Otherwise at
least one positive-degree coefficient is nonzero; subtracting $x$
changes only the constant coefficient, so $P_i(t;W)-x$ remains
nonzero and nonconstant. For fixed $w$, multiplication by each
field constant $w^h$ is binary linear. Summing these linear images
of the affine coefficient maps proves affinity in $t$.
The final assertion follows by applying the root-list lemma to
$P_i(t;W)-x$.
\end{proof}

Two features of \eqref{eq:residualpoly} matter. First, the compiler
needs $P_i(t;w)$ to be affine in the tag $t$, not in $w$. The map
$w\mapsto P_i(t;w)$ is in general neither affine nor injective.
Injectivity is replaced by the root list: if a source has a unique
witness, testing every listed candidate against the source predicate
singles out the one allowed block value, even when other candidates
collide with the witness. Second, the powers start at $W^1$. If the
residual were stored as a constant term, $P_i(t;W)-x$ would vanish
identically for one value of $x$, and the root list would say nothing
about the source.

\subsection{The compiled circuit}

Use ordered witness variables
\begin{equation}\label{eq:variables}
 z=(t,x_1,\ldots,x_m),\qquad t\in\F_2^d,\quad x_i\in\K.
\end{equation}
There are exactly $n=d+mq$ variables. Given a canonically sorted
pool $Q$, define the $i$th block test as follows:
\begin{enumerate}
\item If all residual coefficients are zero, accept precisely when
      $x_i=0$.
\item Otherwise list the roots of $P_i(t;W)-x_i$. Accept precisely
      when some valid root satisfies the normalized predicate $E_i$.
\end{enumerate}
Use $D_i$ root slots, validity bits, and $D_i$ copies of $E_i$ in
the second branch. On the inactive zero-residual branch the
root-list input may be replaced by the fixed polynomial $W$,
so the internal circuit is total on every input. Let $H_Q$ be the
conjunction of the $m$ block tests.

\begin{lemma}[All-YES and one-NO graph identities]\label{lem:graph}
If every source in $Q$ has a unique witness $w_i$, then
\begin{align}
 \Sol(H_Q)&=S_Q=\{\Phi_Q(t):t\in V\},\label{eq:graph}\\
 \Phi_Q(t)&=(t,P_1(t;w_1),\ldots,P_m(t;w_m)).\notag
\end{align}
This is a $d$-dimensional binary affine set, and tag projection
is an affine isomorphism $S_Q\to V$.

If $E_j$ is unsatisfiable and every other source has a unique witness,
then $\Sol(H_Q)$ is a nonempty affine set $R_j(Q)$ obtained by
restricting $t$ to $K_j$, completing the known blocks, and setting
block $j$ to zero. An origin and a full-column-rank basis of $R_j(Q)$
can be computed from $K_j$ and the other $m-1$ witnesses.
For an all-YES pool the same computation, without using $w_j$, gives
\begin{equation}\label{eq:regionidentity}
 R_j(Q)=S_Q\cap\{z:t\in K_j\}.
\end{equation}
\end{lemma}
\begin{proof}
For a unique source and $t\in K_i$, the block is forced to zero,
which equals $P_i(t;w_i)$. If $t\notin K_i$, the root list is exact.
A listed root satisfies $E_i$ if and only if it is $w_i$, so the
block accepts exactly $x_i=P_i(t;w_i)$. This proves
\eqref{eq:graph}. Each block is affine in $t$ by
Lemma~\ref{lem:residual}, and the visible tag coordinates make
$\Phi_Q$ injective and supply its inverse projection. Thus its image
is affine of dimension $d$, including when source witnesses coincide
or equal zero.

If $E_j$ is unsatisfiable, its nonzero-residual branch always rejects.
Its block accepts exactly for $t\in K_j$ and $x_j=0$.
The other blocks have the values already established, giving the
asserted graph over $K_j$. This graph is affine and nonempty
because $K_j$ is nonempty. On that cell $P_j(t;w)=0$ for every
$w$, so the same region occurs in an all-YES graph regardless of
its unknown $j$th witness. This proves \eqref{eq:regionidentity}.

To compute a basis, find $K_j=v+\im(B)$ by Gaussian elimination,
with $B$ full column rank. Complete $v$ by the known block values
and a zero $j$th block. For each column $u$ of $B$, complete $v+u$
in the same way and subtract the completed origin.
Affinity proves that these differences span the region's direction
space. Their tag projections are the independent columns of $B$,
so the completed columns remain independent. If $K_j$ is a
singleton, there are no basis columns and this procedure still
returns its one completed point.
\end{proof}

\subsection{Exact checking in a known region}

\begin{lemma}[Regional singleton checking]\label{lem:checking}
Let $R=v+\im(B)\subseteq\F_2^n$, where $B$ has $r$ independent
columns. Given an explicit affine filter $Az=b$, the intersection
$R\cap\{z:Az=b\}$ is a singleton exactly when
\begin{equation}\label{eq:rankcheck}
 ABu=b-Av
\end{equation}
is consistent and $\rank(AB)=r$. This condition is decidable in
polynomial time from the explicit data.
\end{lemma}
\begin{proof}
The map $u\mapsto v+Bu$ is a bijection from $\F_2^r$ to $R$.
The retained points therefore correspond bijectively to solutions
of \eqref{eq:rankcheck}. A consistent linear system on $r$ variables
has $2^{r-\rank(AB)}$ solutions. Gaussian elimination tests
consistency and rank. The same formula covers $r=0$: consistency
then means that the single point $v$ satisfies the filter.
\end{proof}

\section{Size of the compiled circuit}\label{sec:size}

We bound the size of $H_Q$ by a polynomial in $N=n+\ell+2$ whose
degree does not depend on $k$. This uniformity is what allows $k$ to
grow in Section~\ref{sec:growing}.

\begin{proposition}[One polynomial bound for all parameters]
\label{prop:size}
Let $\ell\ge1$, $1\le q\le\ell$, and $k\ge1$ be integers, and use
\eqref{eq:parameters}. Set $N=n+\ell+2$. There is an absolute
constant $C_0$ such that the compiler, its output gate count and
description length, and the known-region construction are bounded
by $N^{C_0}$. All pools at these parameters have one common padded
length $L_{\ell,q,k}\le N^{C_0}$, and all designations of one pool
give the identical description. The exponent $C_0$ is independent
of $\ell,q,k$. Witness arity is exactly $(q+1)^{k+1}-1$.
\end{proposition}
\begin{proof}
For the flag regions $r_i\le d=m-1$, so
\begin{equation}\label{eq:degree}
 D=\max_iD_i=\left\lceil\frac d q\right\rceil\le d\le n.
\end{equation}
Also $m,d,q\le n+1\le N$. Computing all residuals and packed
coefficients is polynomial in $m,d,q$. The root-list algorithm in
Lemma~\ref{lem:roots} has one fixed polynomial bound in $D,q$;
none of its loop bounds depends on $k$. Each block uses at most
$D$ copies of an $O(\ell+q)$-gate normalized source predicate.
Thus the total gate count is bounded by one fixed polynomial in
$m,d,q,D,\ell$, and therefore by $N^{C_0}$ for an absolute
constant after enlarging it. Gate-index encoding adds only a
logarithmic factor, which is absorbed by the same choice of $C_0$.
Gaussian elimination and evaluation of the known block maps give
the same kind of fixed polynomial bound for region recovery.
All roots and intermediate algebraic data remain computed wires;
the only witness inputs are those in \eqref{eq:variables}.

Fix, once per parameter triple, the field modulus, basis, flag
equations, grouping order, normalization procedure, slot order,
root-list circuit, wire order, and syntax of all operations.
Sort the source descriptions before placing their predicates into
blocks. Pad every output to one prescribed length bounded by the
fixed polynomial above. The compiled bytes depend on the pool
and these parameter-dependent choices alone. Neither the
unknown-member designation nor any source witness is an input
to the compiler. Witnesses enter only the separate known-region
construction in Lemma~\ref{lem:graph}.
\end{proof}

With $q=\ell$ and $k$ fixed, $N$ is polynomial in the source length
$\ell$. When $k$ grows, the exponent $C_0$ stays the same, and the
additional cost appears only through the larger arity $n$.

For $k=1$ we have $d=q$, $m=q+1$, and $D=1$. Each residual polynomial
is then a single field multiplication of the source witness, and the
threshold $2/m$ is about $2/\sqrt n$. For $k=2$,
\[
 d=q^2+2q,\qquad m=(q+1)^2,\qquad D=q+2,
 \qquad n=(q+1)^3-1,
\]
so the first improvement of the exponent already needs root lists of
length linear in $q$. Larger $k$ increases the verification cost and
keeps one $q$-bit witness block per source.

The canonical description matters because the isolator may depend on
syntax: it could output unrelated filters on two encodings of the
same function. In the argument, all designations of a pool see
literally the same $H_Q$, the same length-$L_\ell$ advice, and later
the same seed list. No invariance under equivalent descriptions is
needed.

\section{Fixing the seeds and averaging over designations}\label{sec:seeds}

We first replace the isolator's randomness by one fixed list of seeds
that works for all inputs of a given length. We then bound, for an
all-YES pool, how often the resulting filters can isolate inside the
regions $R_j(Q)$ on average over $j$.

\begin{lemma}[Simultaneous seed fixing]\label{lem:seeds}
Fix an input description length $L$ and a randomized transformation
with random-bit budget $r(L)$. On each of any collection of at most
$2^L$ descriptions, let a seed define a Boolean success event with
probability $\mu_H$. For $\gamma>0$, there is a list of
\[
 J=O\bigl((L+1)\gamma^{-2}\bigr)
\]
seeds for which every description in that collection has empirical
success frequency at least $\mu_H-\gamma$. If $L$, $r(L)$, and
$\gamma^{-1}$ are polynomial in a parameter, so is the list's total
bit length. Repeated seeds are allowed.
\end{lemma}
\begin{proof}
For a fixed description, draw $J$ independent uniform seeds and
let $X_a\in\{0,1\}$ be the corresponding success bits, with mean
$\mu$. We include the exponential estimate. Put
$\phi(t)=\ln\mathbb E e^{t(\mu-X_a)}$.
Then $\phi(0)=\phi'(0)=0$, and $\phi''(t)$ is the variance of a
Bernoulli random variable under the tilted distribution, at most
$1/4$. Integrating twice gives $\phi(t)\le t^2/8$ for $t\ge0$.
Independence and Markov's inequality imply
\[
 \Pr\left[\frac1J\sum_aX_a<\mu-\gamma\right]
 \le \exp\left(-tJ\gamma+\frac{Jt^2}{8}\right).
\]
Taking $t=4\gamma$ gives $e^{-2J\gamma^2}$.
A union bound over at most $2^L$ descriptions is less than one if
$J>(L+1)\ln(2)/(2\gamma^2)$. Some seed list therefore has the
asserted simultaneous property. It occupies $Jr(L)$ bits.
\end{proof}

The lemma is an existence statement for nonuniform advice. It does
not require an efficient test of the success event on the original
input, or an efficient search for the seed list.

For a fixed list of seeds $\sigma_1,\ldots,\sigma_J$, let
$F_{Q,a}$ be the explicit filter produced on $H_Q$ with seed
$\sigma_a$. Define the empirical regional frequency
\begin{equation}\label{eq:regionalfrequency}
 a_j(Q)=\frac1J\sum_{a=1}^J
        \one\{|R_j(Q)\cap F_{Q,a}|=1\}.
\end{equation}
When all witnesses except the $j$th are known, this rational number
is exactly computable using Lemma~\ref{lem:checking}.

\begin{lemma}[Coupling across designations]\label{lem:coupling}
For an all-YES pool $Q$, any common fixed seed list satisfies
\begin{equation}\label{eq:coupling}
 \frac1m\sum_{j=1}^m a_j(Q)\le\frac2m.
\end{equation}
\end{lemma}
\begin{proof}
For one fixed seed, $S_Q\cap F_{Q,a}$ is affine or empty.
Its tag projection $T_a\subseteq V$ is affine or empty, and the
affine graph isomorphism gives, for every $j$,
\[
 |R_j(Q)\cap F_{Q,a}|=|T_a\cap K_j|.
\]
Lemma~\ref{lem:flag} bounds the number of singleton intersections
by two. Average first over $j$ and then over the common seeds.
\end{proof}

The bound \eqref{eq:coupling} uses the fact that, for each seed, one
filter is shared by all designations. This holds because the
description $H_Q$ does not depend on the designation. The bound would
not follow with designation-dependent descriptions, advice, or seed
lists.

\section{From the isolator to circuits for unique satisfiability}\label{sec:domination}

We now compare the isolator's success with the bound $2/m$ of
Lemma~\ref{lem:coupling}, build the testers of the overview, and
combine them into a decider for $\USAT(\ell,q)$.

Fix any parameter triple $(\ell,q,k)$ from \eqref{eq:parameters}, and let
$\alpha=k/(k+1)$. Since $m=(n+1)^\alpha$, the gap between the
hypothesized success and the regional bound is
\begin{align}
 \frac2{n^\alpha}-\frac2m
 &=2\bigl(n^{-\alpha}-(n+1)^{-\alpha}\bigr)\notag\\
 &=2\alpha\int_n^{n+1}x^{-\alpha-1}\,dx
  >2\alpha(n+1)^{-\alpha-1}\notag\\
 &=\frac{2k}{k+1}(q+1)^{-(2k+1)}
  \ge(q+1)^{-(2k+1)}.\label{eq:gap}
\end{align}
The inequality is strict because the integrand is strictly decreasing.
The identity $m=(n+1)^\alpha$ is the reason for the choice $d=m-1$ in
\eqref{eq:parameters}: taking $d=q^k$ would give the same exponent
but not the exact constant $2$.

Set
\begin{equation}\label{eq:threshold}
 B_\ell=(q+1)^{2k+1},\qquad
 \gamma=\frac1{4B_\ell},\qquad
 \tau=\frac2m+\frac1{2B_\ell}.
\end{equation}
Apply Lemma~\ref{lem:seeds} to singleton success on every nonempty
affine circuit description of the common length $L_\ell$.
The bounds below show that this supplies one seed list of size
polynomial in $N=n+\ell+2$, with degree independent of $k$. In particular, for a
one-NO pool the entire satisfying set is $R_j(Q)$, a nonempty
affine set by Lemma~\ref{lem:graph}, so
\begin{equation}\label{eq:oneNO}
 a_j(Q)\ge\frac2{n^\alpha}-\gamma
       >\frac2m+\frac3{4B_\ell}>\tau.
\end{equation}
For every all-YES pool, Lemma~\ref{lem:coupling} gives an average
at most $2/m<\tau$. All comparisons with $\tau$ use rational numbers
of polynomial bit length; the testers never approximate
$n^{-\alpha}$ numerically.

We now state the combinatorial step in a form that will also apply
to explicit unions.

\begin{lemma}[Domination by sound testers]\label{lem:domination}
Suppose that, for every set $X$ of $m-1$ distinct YES source
descriptions and their witnesses, there is a tester
$T_X$ with the following properties:
\begin{enumerate}
\item $T_X$ rejects every NO source description;
\item for every all-YES pool $Q$ of size $m$, at least one
$y\in Q$ is accepted by $T_{Q\setminus\{y\}}$.
\end{enumerate}
There is a nonuniform promise decider using $O(m(\ell+1))$ such
testers and fewer than $2m$ hardwired YES descriptions. In particular,
if all resources are polynomial in a parameter $N\ge\ell+m$,
then so are the decider size and advice.
\end{lemma}
\begin{proof}
Let $S$ be a set of $M\ge2m$ currently uncovered YES descriptions.
Every $m$-subset $Q\subseteq S$ supplies at least one accepted pair
$(X,y)$ with $X=Q\setminus\{y\}$ and $y\in S\setminus X$.
Distinct pools give distinct pairs. Averaging over the
$\binom{M}{m-1}$ possible sets $X$ shows that some tester accepts
at least
\[
 \frac{\binom Mm}{\binom M{m-1}}
 =\frac{M-m+1}{m}\ge\frac{M}{2m}
\]
members of $S$ outside $X$. Choose that tester as nonuniform advice
and remove every member of $S$ that it accepts.
Each such round decreases the number remaining by a factor at most
$1-1/(2m)$. Starting with at most $2^\ell$ descriptions, after
$O(m(\ell+1))$ rounds fewer than $2m$ remain: use
$(1-1/(2m))^r\le e^{-r/(2m)}$.
Hardwire those residual YES descriptions individually. If there were
fewer than $2m$ initially, hardwire them without any rounds.

The OR of the chosen testers and residual membership tests accepts
every YES description. Each tester is sound on NO inputs, as are
the membership tests, so the OR rejects every NO description.
The number of components is $O(m(\ell+1))$; their total size
is polynomial in $N$ under the stated assumption.
The chosen sets and their witnesses need not be efficiently found;
they are nonuniform advice.
\end{proof}

\begin{theorem}[Uniform-resource isolation transfer]\label{thm:resource}
For every $\ell\ge1$, $1\le q\le\ell$, and $k\ge1$, put
$m=(q+1)^k$, $n=(q+1)^{k+1}-1$, and $N=n+\ell+2$.
Fix an isolator whose circuit size, advice, seed length, and output
length are bounded by $(L+2)^a$ for some fixed $a$.
If its success on every nonempty affine circuit of the common
compiled length and arity is at least $2/n^{k/(k+1)}$, then
$\USAT(\ell,q)$ has a nonuniform decider of size and advice
$N^{\kappa(a)}$. Here $\kappa(a)\ge1$ depends only on $a$ and the
fixed encodings and algorithms, and is independent of $\ell,q,k$.
One such exponent is fixed for the remainder of the paper.
\end{theorem}
\begin{proof}
The gap computation \eqref{eq:gap} holds for every $k\ge1$.
Its denominator satisfies the uniform bound
\begin{equation}\label{eq:Bbudget}
 B_\ell=(q+1)^{2k+1}=m(n+1)\le(n+1)^2\le N^2.
\end{equation}
By Proposition~\ref{prop:size}, the common description length is
$L\le N^{C_0}$. Lemma~\ref{lem:seeds} with
$\gamma=1/(4B_\ell)$ uses
$J=O((L+1)B_\ell^2)$ seeds. Since each seed has length at most
$(L+2)^a$, their total length is bounded by a fixed polynomial
in $N$, independently of $k$. The comparison threshold is rational:
\[
 \tau=\frac2m+\frac1{2B_\ell}=\frac{4n+5}{2B_\ell},
\]
so exact integer comparisons suffice. Its bit length is $O(\log N)$.

Fix $m-1$ distinct YES descriptions $X$ and their witnesses.
The tester $T_X$ accepts immediately if the input $y$ is in $X$.
Otherwise form the canonical sorted pool $Q=X\cup\{y\}$ and
find the position $j$ of $y$. Compile $H_Q$, recover $R_j(Q)$
from the known witnesses, and compute $a_j(Q)$ with the common
seed list. Accept exactly when $a_j(Q)<\tau$.
For NO $y$, this is a one-NO pool, so \eqref{eq:oneNO} forces
rejection. For an all-YES pool, \eqref{eq:coupling} and $2/m<\tau$
ensure that some designation is accepted.

Compilation and region recovery have fixed polynomial size in $N$.
The isolator computation on the compiled description, its explicit
output, and every regional Gaussian-elimination check also have
fixed polynomial size in $N$. Repeating for all $J$ seeds and
performing exact counting preserves that property. The advice for
$X$ and its witnesses uses $O(m(\ell+q))$ bits. All parameter data,
including the field modulus, equations, seed list, and isolator
advice, have the same fixed polynomial bound.
Lemma~\ref{lem:domination} combines $O(m(\ell+1))=O(N^2)$ testers
and fewer than $2m$ residual membership tests. Their total size
and advice are $N^{\kappa(a)}$ for one exponent independent of $k$.

Descriptions outside the witness cap or malformed descriptions can
be rejected syntactically. If there are too few YES descriptions
to form a pool, the hardwired residual case of domination already
supplies the bound. No success promise is used on an unsatisfiable
compiled circuit: each one-NO gadget has a nonempty affine
satisfying set.
\end{proof}

\begin{corollary}[The fixed-parameter promise consequence]
\label{prop:promise}
Under Theorem~\ref{thm:main}'s hypothesis, $\USAT$ has
polynomial-size nonuniform deciders.
\end{corollary}
\begin{proof}
Take $q=\ell$ and the fixed $k$ of that theorem. Then
$N=\ell^{O(k+1)}$, so Theorem~\ref{thm:resource} is polynomial
in $\ell$. For sufficiently large $\ell$ the arity lies in the
promised regime; finitely many smaller lengths may be hardwired.
\end{proof}

\section{Witness-preserving transfer from unique satisfiability}
\label{sec:VV}

The last step is the reduction from satisfiability to unique
satisfiability \cite{VV}. We need a version that fixes the randomness
nonuniformly and keeps the number of witness variables unchanged,
because Section~\ref{sec:growing} measures circuit size in that
number.

\begin{lemma}[Quantitative nonuniform Valiant--Vazirani transfer]
\label{lem:VV}
For each description length $s$ and witness cap $1\le v\le s$,
there is a common query length
$\ell_{s,v}=O(s^2\log(s+2))$ with the following property.
If $\USAT(\ell_{s,v},v)$ has a nonuniform decider of size $T$, then
circuit satisfiability on length-$s$ descriptions with at most $v$
witness variables has circuits of size $(s+2)^{O(1)}(T+1)$.
The query circuits have the same $v$ witness variables; all their
gates are computed wires. In particular, polynomial-size deciders
for $\USAT$ imply $\NP\subseteq\Ppoly$.
\end{lemma}
\begin{proof}
Normalize an input circuit $C$ with at most $v$ witness variables
to exactly $v$ by forcing added variables to zero. This adds $O(v)$
gates and preserves the number of satisfying assignments.
If $C$ is satisfiable, let that number be $M\ge1$.
For $r\ge1$, choose a uniform binary $r\times v$ matrix $A$
and an independent uniform target $b\in\F_2^r$. Let $Z$ count
satisfying assignments retained by $Ax=b$, and put $\lambda=M2^{-r}$.
Each assignment survives with probability $2^{-r}$, and two distinct
assignments survive jointly with probability $2^{-2r}$.
Indeed, in each row their nonzero difference has a uniform parity,
and the independent target bit supplies the other independent bit;
the rows are independent. Thus
\[
 \mathbb E Z=\lambda,\qquad
 \mathbb E[Z(Z-1)]=M(M-1)2^{-2r}\le\lambda^2.
\]
For every nonnegative integer $z$,
$\one\{z=1\}\ge z-z(z-1)$, so
\begin{equation}\label{eq:VVprob}
 \Pr[Z=1]\ge\lambda-\lambda^2.
\end{equation}
For $r=\lceil\log_2M\rceil+1$ we have
$1/4<\lambda\le1/2$ and $1\le r\le v+1$.
The right side of \eqref{eq:VVprob} is then at least $3/16$.

Use $\lceil\log_2(v+1)\rceil$ fair bits to choose among
$r=1,\ldots,v+1$, sending unused bit patterns to an unsatisfiable
pruning. The good value occurs with probability at least
$1/[2(v+1)]$. A fixed additional random-bit budget describes the
largest possible matrix and target, with unused rows ignored.
One trial outputs $C'(x)=C(x)\wedge[Ax=b]$, and its singleton
probability on every satisfiable input is at least
\[
 p_v=\frac3{32(v+1)}.
\]
Every trial on an unsatisfiable input remains unsatisfiable.
A dense matrix supplies $O(v^2)$ parity gates and coefficients,
and does not add any witness variables. With indexed gate
encoding, all query descriptions have length
$O((s+v^2)\log(s+v+2))=O(s^2\log(s+2))$.
Pad them to one common length $\ell_{s,v}$, chosen at least $v$.

Take an integer $J>(s+1)\ln(2)/p_v$. Then
$J=O((s+1)(v+1))$. For a fixed satisfiable length-$s$ input,
the probability none of $J$ independent trials isolates is at most
$(1-p_v)^J\le e^{-p_vJ}<2^{-s-1}$.
There are at most $2^s$ input descriptions, so a union bound fixes
one common list of trials that isolates every satisfiable input
within the witness cap at least once. The list has polynomial
bit length in $s$.

Hardwire that list and the decider for $\USAT(\ell_{s,v},v)$.
Build the prunings, run the decider on each, and take the OR.
An unsatisfiable input produces only promised NO queries; a
satisfiable input produces at least one promised YES query.
The decider's arbitrary answers on multiple-witness queries cannot
invalidate either conclusion. Query construction, the seed list,
and $J$ decider copies give size $(s+2)^{O(1)}(T+1)$.
All advice is chosen once for the pair $(s,v)$.

For the polynomial-size conclusion take $v=s$. For an arbitrary
language in $\NP$, unroll its polynomial-time verifier, with the
input fixed, into a circuit on its certificate bits. The input is
in the language exactly when that circuit is satisfiable.
The construction and resulting description length are polynomial
in the input length, so composing with polynomial-size circuit
satisfiability deciders gives $\NP\subseteq\Ppoly$.
\end{proof}

\begin{proof}[Completion of Theorem~\ref{thm:main}]
Corollary~\ref{prop:promise} gives polynomial-size promise deciders,
and Lemma~\ref{lem:VV} supplies the transfer to $\NP\subseteq\Ppoly$.
\end{proof}

The hashing argument in this proof also shows that a random explicit
affine filter isolates with probability $\Omega(1/n)$ on every
nonempty set of $n$-bit witnesses, with no knowledge of the set. In
particular it needs no basis for an affine input set. So success of
order $1/n$ is attainable, and all our hypotheses lie above it.

\section{Success probabilities closer to \texorpdfstring{$1/n$}{1/n}}
\label{sec:growing}

Theorem~\ref{thm:resource} holds for every $k$ with one resource
exponent, so $k$ may grow with $q$. This section chooses $k$ as a
function of $q$ and proves Theorem~\ref{thm:growingintro}.
For any schedule $k=k(q)\ge1$, the exact threshold is
\begin{equation}\label{eq:endpointnumerator}
 \frac2{n^{k/(k+1)}}=\frac{2n^{1/(k+1)}}n
 <\frac{2(q+1)}n,
 \qquad n=(q+1)^{k+1}-1.
\end{equation}
Let $h=\log_2(q+1)$. The exact arity also gives
\begin{equation}\label{eq:logbudget}
 kh<\log_2n<(k+1)h.
\end{equation}
Thus, when $\ell=q^{O(1)}$, the promise decider of
Theorem~\ref{thm:resource} has all its resources bounded by
\[
 2^{O((k+1)\log(q+1))},
\]
with a constant independent of $k$.
For $k=\lceil B\log_2(q+1)\rceil$ and fixed $B>0$, this is
$2^{O((\log q)^2)}$. For $k=\lceil q^\varepsilon\rceil$ and
fixed $0<\varepsilon<1$, it is
$2^{O(q^\varepsilon\log q)}=2^{o(q)}$.
The seed list, threshold comparison, regional checks, and domination
advice are all included in these bounds.

\begin{proof}[Proof of Theorem~\ref{thm:growingintro}, part 1]
Fix $\theta,c>0$ from the success hypothesis and choose a fixed
$B>1/\theta^2$. Set $k=\lceil B\log_2(q+1)\rceil$.
By \eqref{eq:logbudget},
\[
 \sqrt{\log_2n}>\sqrt B\,\log_2(q+1),
 \qquad
 c\,2^{\theta\sqrt{\log_2n}}>c(q+1)^{\theta\sqrt B}.
\]
Since $\theta\sqrt B>1$, this exceeds $2(q+1)$ for all sufficiently
large $q$. Equations \eqref{eq:endpointnumerator} and the hypothesis
therefore give the success needed in Theorem~\ref{thm:resource}.

For a length-$s$ circuit satisfiability input, take $q=s$ and
$\ell=\ell_{s,s}=s^{O(1)}$ from Lemma~\ref{lem:VV}.
The restricted promise decider and its transfer have size
$2^{O((\log s)^2)}$.
For a language in $\NP$, a fixed reduction has output description
length at most $s^d$ on inputs of length $s$, for some fixed $d$.
Substituting this length changes $(\log s)^2$ by only a constant
factor. Every language in $\NP$ thus has quasipolynomial-size
circuits. Finite exceptional lengths may be hardwired.
\end{proof}

\begin{theorem}[Witness-sensitive polylogarithmic consequence]
\label{thm:polylog}
Assume the explicit affine isolator has success at least
$c(\log_2 n)^A/n$ on every sufficiently large nonempty affine
input, with fixed $c>0$ and $A>1$.
For each $1\le v\le s$, circuit satisfiability on length-$s$
descriptions with at most $v$ witness variables has nonuniform
circuits of size
\[
 (s+2)^{O(1)}+2^{O(v^{1/A})}.
\]
In particular, taking $v=s$ proves part 2 of
Theorem~\ref{thm:growingintro}.
\end{theorem}
\begin{proof}
For the source witness cap $q$, define
\begin{equation}\label{eq:optimizedk}
 t=\left(\frac{4(q+1)}c\right)^{1/A},\qquad
 k=\max\left(1,\left\lceil\frac{t}{\log_2(q+1)}\right\rceil\right).
\end{equation}
Then $k\log_2(q+1)\ge t$, so \eqref{eq:logbudget} gives
\[
 c(\log_2 n)^A>ct^A=4(q+1)>2(q+1).
\]
The endpoint requirement follows from \eqref{eq:endpointnumerator}.
Also $k\le t/\log_2(q+1)+1$, since the quotient is positive,
and hence
\[
 \log_2n<t+2\log_2(q+1)=O(q^{1/A}).
\]
Here and below the constants may depend on $c,A$.
Theorem~\ref{thm:resource} consequently gives size
$(\ell+2)^{O(1)}+2^{O(q^{1/A})}$: for a fixed $C$,
$(n+\ell+2)^C$ is at most a constant multiple of
$(n+1)^C+(\ell+1)^C$.

Now apply Lemma~\ref{lem:VV} with $q=v$ and
$\ell=\ell_{s,v}=s^{O(1)}$. Multiplying by its polynomial
transfer overhead still gives
$(s+2)^{O(1)}+2^{O(v^{1/A})}$, using $xy\le x^2+y^2$
for nonnegative $x,y$ to absorb the mixed product.
For the finitely many small witness caps whose compiled arities
fall below the isolator's promised range, direct exhaustive
checking costs only a constant times a polynomial in $s$.
When $v=s$, the bound is $2^{O(s^{1/A})}=2^{o(s)}$.
\end{proof}

The choice $k=\lceil q^\varepsilon\rceil$ is another valid
subexponential schedule whenever $0<\varepsilon<1$ and the
success numerator dominates $2(q+1)$. Under the hypothesis of
Theorem~\ref{thm:polylog}, this holds for every
$1/A\le\varepsilon<1$, since
$(\log n)^A>q^{\varepsilon A}(\log(q+1))^A$ eventually dominates
$q$. The optimized choice \eqref{eq:optimizedk} removes the extra
$\log q$ from the resulting size exponent.

\begin{theorem}[Quantitative consequence of a logarithmic numerator]
\label{thm:constantlog}
Fix $C>0$ and an explicit affine isolator whose computation, advice,
seed length, and output length are bounded by $(L+2)^a$.
Suppose its success is at least $C\log_2(n)/n$ on every sufficiently
large nonempty affine circuit input. For each $1\le v\le s$,
circuit satisfiability on length-$s$ descriptions with at most $v$
witness variables has nonuniform circuits of size
\begin{equation}\label{eq:constantlogsize}
 K_{\Iso,C}(s+2)^{b(a)}\,2^{(2\kappa(a)/C)v}.
\end{equation}
Here $\kappa(a)$ is as in Theorem~\ref{thm:resource}, the exponent
$b(a)$ depends only on $a$ and the fixed algorithms and encodings,
and $K_{\Iso,C}$ is a constant. This proves part 4 of
Theorem~\ref{thm:growingintro}.
\end{theorem}
\begin{proof}
For a source witness cap $q$, put $h=\log_2(q+1)$ and choose
\begin{equation}\label{eq:constantlogk}
 k=\max\left(1,\left\lceil\frac{2(q+1)}{Ch}\right\rceil\right),
 \qquad n=(q+1)^{k+1}-1.
\end{equation}
The ceiling is already at least one. By \eqref{eq:logbudget},
\begin{equation}\label{eq:constantlogbudget}
 \frac{2(q+1)}C\le kh<\log_2n
 <(k+1)h\le\frac{2(q+1)}C+2h.
\end{equation}
Consequently
\[
 C\log_2n>2(q+1)>2n^{1/(k+1)}.
\]
The hypothesis therefore exceeds the exact threshold in
\eqref{eq:endpointnumerator}, and Theorem~\ref{thm:resource}
applies with its existing inverse-polynomial gap. No new precision
or seed bound is required.

Write $\kappa=\kappa(a)\ge1$. The upper bound in
\eqref{eq:constantlogbudget} gives
\[
 n^\kappa<2^{2\kappa/C}(q+1)^{2\kappa}
                 2^{(2\kappa/C)q}.
\]
Since $q\le\ell$ and
\[
 (n+\ell+2)^\kappa
 \le2^{\kappa-1}\bigl(n^\kappa+(\ell+2)^\kappa\bigr),
\]
the restricted promise decider has size at most
$K_C(\ell+2)^{2\kappa}2^{(2\kappa/C)q}$ for a constant $K_C$.
Apply Lemma~\ref{lem:VV} with $q=v$ and
$\ell=\ell_{s,v}=O(s^2\log(s+2))$. Its polynomial overhead
multiplies this expression and is absorbed into $(s+2)^{b(a)}$.
In particular, we keep that overhead multiplicative instead of
using a product-to-sum estimate that would enlarge the exponential
coefficient. This proves \eqref{eq:constantlogsize} with the exact
coefficient $2\kappa(a)/C$.
For the bounded exceptional witness caps whose compiled arities
lie below the isolator's eventual guarantee, exhaustive evaluation
costs a constant times a polynomial in $s$, and is absorbed by
$K_{\Iso,C}$.
\end{proof}

\begin{corollary}[Resource-qualified exponential rates]
\label{cor:rates}
Fix a resource degree $a$ and a target $\varepsilon>0$.
If an isolator with that degree has eventual success at least
$C\log_2(n)/n$ for some $C\ge2\kappa(a)/\varepsilon$, then circuit
satisfiability has size
$(s+2)^{O(1)}2^{\varepsilon v}$ at witness cap $v$.
\end{corollary}
\begin{proof}
Substitute $2\kappa(a)/C\le\varepsilon$ into
\eqref{eq:constantlogsize}.
\end{proof}

The order of the quantifiers matters in Corollary~\ref{cor:rates}.
Isolators for arbitrarily large constants $C$ with one common resource
degree $a$ give arbitrarily small exponential rates. Isolators for
every $C$ with uncontrolled degrees $a(C)$ give only the rate
$2\kappa(a(C))/C$, which need not tend to zero. A single fixed
$C\log(n)/n$ guarantee gives the rate $2\kappa(a)/C$, and not by
itself a subexponential bound.

\begin{corollary}[Every superlogarithmic numerator]\label{cor:superlog}
In the explicit affine model, an isolator of resource degree $a$
with a success guarantee $p(n)=\omega(\log_2n/n)$ on every
sufficiently large nonempty affine input implies circuit
satisfiability circuits of size $(s+2)^{b(a)}2^{o(v)}$ for
length-$s$ descriptions with at most $v$ witness variables.
Here $b(a)$ is as in Theorem~\ref{thm:constantlog}. This proves
part 3 of Theorem~\ref{thm:growingintro}.
\end{corollary}
\begin{proof}
The same degree-$a$ isolator eventually meets $C\log_2(n)/n$ for
every fixed $C>0$. Theorem~\ref{thm:constantlog} gives
$K_C(s+2)^{b(a)}2^{(2\kappa(a)/C)v}$ with $b(a)$ independent of $C$.
Set $C_j=2\kappa(a)j$ and take strictly increasing integer thresholds
$V_j\ge j\max\{1,\log_2K_{C_j}\}$.
For $V_j\le v<V_{j+1}$, the bound is at most
$(s+2)^{b(a)}2^{2v/j}=(s+2)^{b(a)}2^{o(v)}$.
Use the $C_1$ family for $v<V_1$; all choices are nonuniform and
independent of $s$.
\end{proof}

\paragraph{Reduction length.}
For a fixed $\NP$ language whose reduction to circuit satisfiability
has description length $O(s^d)$, Theorem~\ref{thm:polylog} gives
size $2^{O(s^{d/A})}$, with polynomial overhead absorbed.
This is subexponential in that language's own input length when
$d<A$. It does not show that all of $\NP$ has subexponential circuits
for one fixed $A>1$.

\paragraph{Hardness assumptions.}
Viewing a $3$-CNF formula as a circuit, Theorem~\ref{thm:polylog}
yields circuits of size $2^{O(v^{1/A})}$ for satisfiability of
polynomial-length $3$-CNF formulas on $v$ variables, and
Corollary~\ref{cor:superlog} yields size $2^{o(v)}$ under any
$\omega(\log n/n)$ guarantee. Both hypotheses therefore contradict a
nonuniform assumption that rules out $2^{o(v)}$-size circuits for
that problem. For the fixed logarithmic guarantee,
Theorem~\ref{thm:constantlog} gives size $2^{(2\kappa(a)/C+o(1))v}$
on polynomial-length inputs, which contradicts a nonuniform
assumption excluding circuits of size $2^{\eta v}$ once
$2\kappa(a)/C<\eta$. The weaker condition $C>2\kappa(a)$ improves on
exhaustive search but need not contradict a given assumption.
We state these comparisons in the variable count $v$ because the
bit length $s$ can include logarithmic gate-index overhead. The
assumptions have to be nonuniform, since all reductions here fix
seeds and advice nonuniformly; the uniform exponential time
hypothesis does not suffice.

\section{A converse from SAT search to decision}\label{sec:converse}

The converse uses the usual search-to-decision self-reduction of
SAT; see \cite[Section~1.1]{BG} for background on this standard
property. We state the nonuniform size dependence explicitly,
since different resource classes have different closure properties.

\begin{lemma}[SAT circuits produce explicit singleton filters]
\label{lem:converse}
Suppose that, for each description length $t$ and witness cap
$1\le v\le t$, circuit satisfiability has a decision circuit
$D_{t,v}$ of size at most $S(t,v)$. There is a fixed polynomial
padding and query-length bound $P(L,n)\ge n$ such that, for each
pair $(L,n)$ with $1\le n\le L$, a deterministic nonuniform
transformation on length-$L$ circuits with $n$ ordered witness
variables has size
\begin{equation}\label{eq:conversesize}
 n\,S(P(L,n),n)+\operatorname{poly}(L+n)
\end{equation}
and outputs an explicit affine filter isolating its input $H$
whenever $H$ is satisfiable. Its output has $n$ equations, and its success is one
on every nonempty satisfying set, whether affine or not.
\end{lemma}
\begin{proof}
For a prefix $u\in\{0,1\}^i$, let $H_u$ be the circuit
\[
 H_u(z)=H(z)\wedge\bigwedge_{j=1}^i[z_j=u_j].
\]
This query has the same $n$ witness variables as $H$.
All equality tests and conjunctions are computed gates.
Our encoding allows construction and padding of every such query
to one common attainable length $P(L,n)=\operatorname{poly}(L+n)$,
independent of the prefix and its length. Query construction has
polynomial circuit size. The bound may be chosen large enough to
cover every $i\le n$ and every length-$L$ input description.

Starting from the empty prefix, construct $w_1,\ldots,w_n$
successively. At stage $i$, evaluate $D_{P(L,n),n}$ on the query
$H_{w_1\cdots w_{i-1}0}$. If it reports satisfiable, set $w_i=0$;
otherwise set $w_i=1$. When $H$ is satisfiable, induction on $i$
shows that the chosen prefix has a satisfying extension: the
zero extension is used if possible, and otherwise the previously
feasible prefix must have a one extension. Hence the completed
string $w$ satisfies $H$.

Unroll the $n$ adaptive stages into a Boolean circuit. Each stage
uses one copy of the same decision circuit, with its nonuniform
data fixed for $(P(L,n),n)$. The input description and the already
computed prefix supply the query wires. The total cost is
\eqref{eq:conversesize}; the polynomial term includes all query
construction, parsing, and output writing. Emit
\[
 A=I_n,\qquad b=w.
\]
The explicit $n\times n$ identity matrix and right-hand side have
$O(n^2)$ bits, and their solution set is exactly $\{w\}$.
Thus $\Sol(H)\cap\{z:Az=b\}=\{w\}$ whenever $H$ is satisfiable.
No correctness requirement is imposed on an unsatisfiable input;
the same procedure may output any $w$ there. Malformed circuit
descriptions can be detected syntactically and assigned an arbitrary
output. No additional witness variables are introduced at any stage.
\end{proof}

\begin{corollary}[Equivalence at polynomial resources]
\label{cor:equivalence}
The following statements are equivalent.
\begin{enumerate}
\item $\NP\subseteq\Ppoly$.
\item There is a deterministic nonuniform polynomial-size
transformation that returns an explicit affine filter with success
one on every satisfiable circuit input.
\item For some fixed $0\le\delta<1$, there is a randomized
nonuniform polynomial-size explicit affine isolator with success
$\Omega(n^{-\delta})$ on every sufficiently large nonempty affine
circuit input.
\end{enumerate}
All resources in statements 2 and 3 are polynomial in the input
description length, as in Theorem~\ref{thm:main}.
\end{corollary}
\begin{proof}
Statement 1 gives polynomial-size circuits for circuit satisfiability.
Lemma~\ref{lem:converse} and polynomial closure under $P(L,n)$ give
a polynomial-size transformation for each $(L,n)$. To obtain one
isolator circuit for length $L$, multiplex these transformations over
all $1\le n\le L$, selecting the output by the parsed witness arity
and padding its coefficient array. This adds only polynomial overhead;
the case $n=0$ is handled by direct evaluation. Thus statement 2
holds in the original length-indexed model.
Statement 2 implies statement 3 by using no randomness
and choosing $\delta=0$. Statement 3 implies statement 1 by
Corollary~\ref{cor:exponents}.
\end{proof}

Lemma~\ref{lem:converse} also gives converses at larger resources,
with the composition in \eqref{eq:conversesize}. For example, SAT
circuits of size $2^{O((\log(t+2))^r)}$, for a fixed $r\ge1$, give
isolators of size $2^{O((\log(L+2))^r)}$, because $n\le L$ and $P$ is
a fixed polynomial; multiplexing over all $n\le L$ preserves this
bound. A witness-sensitive SAT bound $(t+2)^{O(1)}2^{g(v)}$ with
$0\le g(v)=o(v)$ gives isolators of size $(L+2)^{O(1)}2^{g(n)}$ for
each pair $(L,n)$, since the queries keep the number of witness
variables. Here the isolators have to be indexed by both $L$ and $n$:
multiplexing over all arities at one description length would charge
for the larger witness caps. A subexponential bound stated in the
description length alone must be composed with $P$, and a polynomial
blowup of the query length need not preserve it.

These converses do not turn Theorem~\ref{thm:growingintro} into an
equivalence. Its hypotheses concern polynomial-size isolators, while
quasipolynomial or subexponential SAT circuits give back only
isolators of that larger size. The equivalence proved here is
Corollary~\ref{cor:equivalence}.

\section{Explicit unions and parity decision trees}\label{sec:unions}

Throughout this section take $q=\ell$ and fixed $k$, so the compiled
arity and description length are polynomial in $\ell$.

An explicit union filter is a list of affine systems
$A_a z=b_a$, with the retained filter equal to their union.
The list and all coefficients have polynomial total description
length in $L$. Components may overlap. The component count may
depend on $L,n$, the input syntax, and the seed, subject to a
stated upper bound.

\begin{lemma}[The bounded-union extremum]\label{lem:union}
For the flag partition, every union $T$ of at most $h\ge1$ affine
sets satisfies
\[
 \#\{j:|T\cap K_j|=1\}\le\min(m,2h).
\]
This maximum is exactly $\min(m,2h)$.
\end{lemma}
\begin{proof}
Write $T=\bigcup_{a=1}^h T_a$, allowing empty components.
If $T\cap K_j=\{x\}$, some component contains $x$ and its
intersection with $K_j$ is also exactly $\{x\}$.
Assign the cell to one such component. By Lemma~\ref{lem:flag},
a component can receive at most two cells, proving the bound
$2h$; the bound $m$ is immediate.

For attainment, choose one representative from each of
$u=\min(m,2h)$ distinct cells. Pair these representatives into
binary affine lines and use a singleton if one point remains.
There are $\lceil u/2\rceil\le h$ components. A binary line
contains only its two endpoints, so the union is exactly the
chosen representatives and isolates precisely those $u$ cells.
\end{proof}

\begin{lemma}[Exact checking for explicit unions]\label{lem:unioncheck}
Whether a known affine region meets an explicit union filter in
exactly one point is decidable in polynomial time in the explicit
input and output lengths.
\end{lemma}
\begin{proof}
Intersect the known region with each component by Gaussian
elimination. If any nonempty intersection has positive dimension,
it already contains at least two points, so the whole union is
not a singleton. Otherwise every nonempty component intersection
is a computed singleton. Deduplicate those points and accept
exactly when one distinct point remains. An empty list rejects.
These operations are polynomial in the number and sizes of the
explicit systems and handle all overlaps.
\end{proof}

For an all-YES graph, the intersection with each component has an
affine tag image. Lemma~\ref{lem:union} therefore replaces the
right side of \eqref{eq:coupling} by
\begin{equation}\label{eq:unionbound}
 b_\ell=\min\left(1,\frac{2h(L_\ell,n_\ell)}m\right),
\end{equation}
where $h(L,n)$ bounds the output component count at those parameters.

\begin{theorem}[Quantitative union separation]\label{thm:uniongeneral}
Fix $k\ge1$ and use the compiled lengths and arities above.
Suppose a randomized nonuniform polynomial-size explicit-union
isolator has singleton success at least $p_\ell$ on every
nonempty affine input of those compiled lengths and arities, where
\[
 p_\ell-b_\ell\ge\varepsilon_\ell>0,
 \qquad \varepsilon_\ell^{-1}=\ell^{O(1)}.
\]
Then $\NP\subseteq\Ppoly$.
\end{theorem}
\begin{proof}
Fix one common seed list with error at most
$\varepsilon_\ell/4$ using Lemma~\ref{lem:seeds}.
Every one-NO regional empirical frequency is at least
$b_\ell+3\varepsilon_\ell/4$, while the all-YES average is
at most $b_\ell$ by Lemma~\ref{lem:union}.
Choose a rational comparison threshold between
$b_\ell+\varepsilon_\ell/4$ and
$b_\ell+\varepsilon_\ell/2$.
A polynomial-bit dyadic number suffices: a grid with spacing
less than one eighth of a fixed inverse-polynomial lower bound
on $\varepsilon_\ell$ has a point in this interval.
The threshold may be length-dependent advice.

Construct the testers as in Theorem~\ref{thm:resource}, now
using Lemma~\ref{lem:unioncheck}. They reject every NO source,
and every all-YES pool supplies an accepted designation.
Lemma~\ref{lem:domination} gives the promise decider, followed
by Lemma~\ref{lem:VV}. All resources remain polynomial.
\end{proof}

\begin{theorem}[Fixed-exponent explicit-union barrier]\label{thm:unions}
Fix $k\ge1$ and $\alpha=k/(k+1)$. Let $h(n)\ge1$ be an
integer-valued arity bound on the number of affine components of
an explicit-union isolator in the same resource model.
The success guarantee
\begin{equation}\label{eq:unionpromise}
 \frac{2h(n)}{n^\alpha}
\end{equation}
on every sufficiently large nonempty affine circuit input implies
$\NP\subseteq\Ppoly$ whenever the hypothesis is nonvacuous.
In particular, it is an eventually feasible numerical requirement
when $2h(n)\le n^\alpha$ for all sufficiently large $n$.
\end{theorem}
\begin{proof}
Write $h=h(n_\ell)$. Nonvacuity forces $2h\le n_\ell^\alpha<m$,
so \eqref{eq:unionbound} is $b_\ell=2h/m$.
Multiplying \eqref{eq:gap} by $h\ge1$ gives
\[
 \frac{2h}{n_\ell^\alpha}-\frac{2h}{m}
 >\frac{h}{B_\ell}\ge\frac1{B_\ell}.
\]
Theorem~\ref{thm:uniongeneral} applies. More explicitly, the
single-component choice $\gamma=1/(4B_\ell)$ and threshold
$2h/m+1/(2B_\ell)$ already separate the two cases.
If \eqref{eq:unionpromise} exceeds one at arbitrarily large arities,
no transformation can satisfy that eventual promise, and the
conditional claim has no nonvacuous application there.
\end{proof}

\begin{corollary}[Component growth versus success exponent]
\label{cor:unions}
Fix $\beta,\delta\ge0$ with $\beta+\delta<1$.
Suppose $h(n)\le n^\beta$ and an explicit union of at most $h(n)$
affine filters has worst-case success $\Omega(n^{-\delta})$ on
nonempty affine circuit inputs. Then $\NP\subseteq\Ppoly$.
The same conclusion holds for explicit parity decision trees of
polynomial total description length with at most $h(n)$ accepting
leaves.
\end{corollary}
\begin{proof}
Choose a fixed $k$ with $\alpha=k/(k+1)>\beta+\delta$.
Then
\[
 \frac{2h(n)}{n^\alpha}\le2n^{\beta-\alpha}
       =o(n^{-\delta}).
\]
A positive constant times $n^{-\delta}$ therefore eventually
dominates \eqref{eq:unionpromise}. The latter bound tends to zero,
so this comparison is nonvacuous at large arities.

Each root-to-accepting-leaf path of a parity decision tree specifies
a finite list of parity values, hence an affine system.
Enumerate accepting paths and discard inconsistent systems.
Their union is exactly the accepted set, with at most one component
per accepting leaf. The list's size is polynomial in the explicit
tree size, even if coefficients along a path are repeated.
Apply the explicit-union result.
\end{proof}

The bound counts accepting leaves, not depth. If the number of
components is bounded in terms of the description length $L$, its
value $h(L_\ell,n_\ell)$ at the compiled parameters is what matters,
and a polynomial in $L$ need not be small relative to $n$. Once
$2h\ge m$, the bound \eqref{eq:unionbound} equals one and no
separation is possible; Lemma~\ref{lem:union} shows that this value
is attained. So filters with an unrestricted polynomial number of
components, and output circuits from which the components cannot be
read off, are not covered.

\section{Limits of the method}\label{sec:barriers}

This section asks how small the threshold can be made by changing
the cells or the parameters, while keeping one block of $q$ bits per
source. It is at least $2q/n$, and the flag construction attains this
up to a factor $1+1/q$.

Consider any indexed sequence of $m\ge2$ nonempty proper affine
regions $\Fam=(K_1,\ldots,K_m)$ in $V=\F_2^d$, allowing
repetitions. Retain the equal-weight statistic \eqref{eq:f} and set
\begin{equation}\label{eq:extrema}
 g=\max_{t\in V}f_{\Fam}(\{t\}),\qquad
 c'=\max_{\substack{T\text{ affine}\\|T|\ne1}}f_{\Fam}(T),
 \qquad c=\max(g,c').
\end{equation}
The empty set may be included in the second maximum.
The ordinary regional averaging argument separates success from $c$.
Using global singleton success also on an all-YES pool gives the
symmetric scalar threshold
\begin{equation}\label{eq:rho}
 \rho=\max\left(g,\frac{c'}{1+c'-g}\right).
\end{equation}
Its denominator is positive: if $g<1$ this is immediate; if $g=1$,
a point lies in all indexed regions, and a two-point line through
it and outside one proper region shows $c'>0$.
To see the threshold, if empirical global singleton frequency on
an all-YES graph is $a\ge u>\rho$, its regional average is at most
\[
 ag+(1-a)c'\le\max\bigl(g,\;ug+(1-u)c'\bigr)<u.
\]
The strict inequalities follow from $g<u$ and
$c'<u(1+c'-g)$. A one-NO pool with global success at least $u$
then has the opposite regional inequality, giving the same
domination separation as before.

\begin{theorem}[Sharp scalar bounds]\label{thm:scalarbarrier}
Every such indexed region family satisfies
\[
 g\ge\frac1m,\qquad c\ge\frac2m,\qquad
 \rho\ge\frac2{m+1}.
\]
The flag family with $d=m-1$ attains both latter bounds exactly:
$g=1/m$, $c'=c=2/m$, and $\rho=2/(m+1)$.
\end{theorem}
\begin{proof}
A point in any nonempty indexed region proves $g\ge1/m$.
Choose two indexed regions. If they intersect, a singleton in
the intersection meets both once. If they are disjoint, the
affine line through one point from each consists of exactly those
two points and meets both regions once. Other regions only
increase the statistic, so $c\ge2/m$ in either case.

If $g\ge2/(m+1)$ then the bound for $\rho$ is immediate.
Otherwise $g<2/(m+1)<2/m$, and $c\ge2/m$ forces $c'\ge2/m$.
For fixed $g<1$, the function $x/(1+x-g)$ is increasing in $x$,
and for fixed positive $x$ it is increasing in $g$.
Consequently
\[
 \rho\ge\frac{c'}{1+c'-g}
 \ge\frac{2/m}{1+2/m-1/m}=\frac2{m+1}.
\]
The flag values follow from Lemma~\ref{lem:flag} and substitution
into \eqref{eq:rho}.
\end{proof}

\begin{corollary}[One-block-per-source arity barrier]
\label{cor:blockbarrier}
Suppose a compiler uses $d$ tag bits and one $q$-bit block per
source, so $n=d+mq$. Its scalar regional thresholds obey
\begin{align*}
 c&\ge\frac2m=\frac{2q}{n-d}\ge\frac{2q}n,\\
 \rho&\ge\frac2{m+1}=\frac{2q}{n-d+q}\ge\frac{2q}{n+q}.
\end{align*}
For the flag parameters $d=m-1$, $n=m(q+1)-1$, the exact values are
\[
 c_{\rm flag}=\frac{2(q+1)}{n+1},\qquad
 \rho_{\rm flag}=\frac{2(q+1)}{n+q+2}.
\]
Thus the flag is optimal at fixed $m$. Relative to the respective
arity-only lower bounds, each flag value is within a factor
$1+1/q$.
\end{corollary}
\begin{proof}
Substitute $m=(n-d)/q$ into Theorem~\ref{thm:scalarbarrier}.
For the flag substitute $m=(n+1)/(q+1)$. The ratios to the
arity-only lower bounds are respectively
\[
 \frac{q+1}{q}\frac n{n+1}<1+\frac1q,
 \qquad
 \frac{q+1}{q}\frac{n+q}{n+q+2}<1+\frac1q.
\]
\end{proof}

These bounds concern the statistic \eqref{eq:f} and compilers that
spend one $q$-bit block per source. They say nothing about other
reductions or about filters that are not affine. Note also that
$2q/n$ is a lower bound in terms of the arity alone. The flag value
is $2(q+1)/(n+1)$, and its ratio to $2q/n$ tends to one as $q$ grows,
even when $k$ grows as well.

The bounds explain the role of $\log n/n$ in
Section~\ref{sec:growing}. If the arity is subexponential,
$n=2^{o(q)}$, then $C\log n=o(q)$ for every fixed $C$, so a guarantee
$C\log n/n$ lies below both floors: $cn\ge2q$, and
$\rho n\ge2qn/(n+q)\ge q$ because $n\ge q$. A fixed logarithmic
numerator therefore cannot be treated with subexponential arity in
this model. Theorem~\ref{thm:constantlog} uses the exponential-arity
choice \eqref{eq:constantlogk} instead, which is why its conclusion
is the rate $2\kappa(a)/C$ and not a subexponential bound. A
different reduction is not excluded.

Finally, the bounds show what the residual polynomials achieve. A
compiler that encodes the cell equations injectively in the block
needs $q\ge\max_i\codim K_i$. For the flag cells this means
$q\ge d=m-1$, hence $n\ge m^2-1$ and a threshold $2/m$ of at least
$2/\sqrt{n+1}$. With root lists the block has width $q$ while the
maximum codimension is $d=(q+1)^k-1$, which exceeds $q$ for $k>1$.
The candidate lists are computed wires. They add to the size of the
circuit, which Theorem~\ref{thm:resource} accounts for, and not to
its arity. What remains is the $q/n$ scale of
Corollary~\ref{cor:blockbarrier}.

\section{Conclusion and open problems}\label{sec:conclusion}

For isolators that output a system of linear equations, reading the
circuit does not help much. Success probability $\Omega(n^{-\delta})$
with $\delta<1$ is equivalent to $\NP\subseteq\Ppoly$, and success
$\omega(\log n/n)$ already gives satisfiability circuits of size
$(s+2)^{O(1)}2^{o(v)}$. The proof hides the dimension of an affine
solution set in the satisfiability of $m$ source instances. Two
ingredients keep the arity of the gadget close to $m$ times the
source witness length: the flag partition, which needs only $m-1$ tag
bits for $m$ cells, and residual polynomials with root lists, which
let a block of $q$ bits serve a cell of codimension far above $q$.

Three questions remain open.
\begin{enumerate}
\item \emph{The endpoint.} Does success $C/n$ for a constant $C$ have
any consequence? With one $q$-bit block per source the threshold is
at least $2q/n$ (Corollary~\ref{cor:blockbarrier}), so the present
method cannot reach it.
\item \emph{Implicit filters and many components.} Do the results
extend to isolators that output a circuit computing an affine
predicate, without its coefficients? Our test needs the coefficients.
For explicit unions, the flag bound is trivial once the number of
components reaches $m/2$ (Lemma~\ref{lem:union}).
\item \emph{Arbitrary outputs.} For isolators whose output is an
arbitrary circuit, the known collapses \cite{DKMW,BD} are far from
$1/n$. Does any inverse-polynomial success probability have a
consequence there?
\end{enumerate}

\end{document}